\documentclass[%
 reprint,
 onecolumn,
 amsmath,
 amssymb,
 aps,
 unsortedaddress
]{revtex4-2}
\usepackage{algpseudocode}
\usepackage{graphicx}
\usepackage{dcolumn}
\usepackage{bm}
\usepackage{subcaption}
\usepackage{makecell}
\usepackage{multirow}
\usepackage{amsmath,amssymb,amsthm,easybmat,verbatim}
\allowdisplaybreaks
\usepackage{caption}
\usepackage{xcolor}
\usepackage{listings}
\usepackage{hyperref}
\usepackage{zx-calculus}
\hypersetup{
    colorlinks,
    linkcolor=blue,
    citecolor=blue,
    urlcolor=blue
}
\definecolor{amethyst}{HTML}{9063CD}

\newtheorem{thm}{Theorem}[section]

\newtheorem{obs}[thm]{Observation}

\newtheorem{construct}[thm]{Construction}

\newtheorem{example}[thm]{Example}

\usepackage{orcidlink}

\usepackage[ruled, linesnumbered]{algorithm2e}
\usepackage{algpseudocode}
\usepackage{algcompatible}
\algnewcommand{\To}{\textbf{To }}
\algnewcommand\Input{\item[\textbf{Input:}]}%
\algnewcommand\Output{\item[\textbf{Output:}]}%

\newlength\myindent
\usepackage{physics}

\usepackage{ulem}

\usepackage{colortbl}
\usepackage{tikz}
\usetikzlibrary{positioning,fit,backgrounds}
\usepackage{tikzit}

\tikzstyle{gate}=[shape=rectangle, text height=1.5ex, text depth=0.25ex, yshift=0.5mm, fill=white, draw=black, minimum height=5mm, yshift=-0.5mm, minimum width=5mm, font={\small}, tikzit category=circuit]
\tikzstyle{big gate}=[shape=rectangle, text height=1.5ex, text depth=0.25ex, yshift=0.5mm, fill=white, draw=black, minimum height=10mm, yshift=-0.5mm, minimum width=5mm, font={\small}, tikzit category=circuit]
\tikzstyle{Z dot}=[inner sep=0mm, minimum size=2mm, shape=circle, draw=black, fill={rgb,255: red,221; green,255; blue,221}, tikzit category=zx]
\tikzstyle{Z phase dot}=[minimum size=5mm, font={\footnotesize\boldmath}, shape=rectangle, rounded corners=2mm, inner sep=0.2mm, outer sep=-2mm, scale=0.8, tikzit shape=circle, draw=black, fill={rgb,255: red,221; green,255; blue,221}, tikzit draw=blue, tikzit category=zx]
\tikzstyle{X dot}=[Z dot, shape=circle, draw=black, fill={rgb,255: red,255; green,136; blue,136}, tikzit category=zx]
\tikzstyle{X phase dot}=[Z phase dot, tikzit shape=circle, tikzit draw=blue, fill={rgb,255: red,255; green,136; blue,136}, font={\footnotesize\boldmath}, tikzit category=zx]
\tikzstyle{hadamard}=[fill=yellow, draw=black, shape=rectangle, inner sep=0.6mm, minimum height=1.5mm, minimum width=1.5mm, tikzit category=zx]
\tikzstyle{paulibox}=[fill={rgb,255: red,221; green,221; blue,255}, draw=black, shape=rectangle, inner sep=0.6mm, minimum height=5mm, minimum width=5mm, font={\footnotesize}, text height=1.5ex, text depth=0.25ex, tikzit category=zx]
\tikzstyle{vertex}=[inner sep=0mm, minimum size=1mm, shape=circle, draw=black, fill=black, tikzit category=misc]
\tikzstyle{vertex set}=[inner sep=0mm, minimum size=1mm, shape=circle, draw=black, fill=white, font={\footnotesize\boldmath}, tikzit category=misc]
\tikzstyle{small black dot}=[fill=black, draw=black, shape=circle, inner sep=0pt, minimum width=1.2mm, tikzit category=circuit]
\tikzstyle{cnot ctrl}=[fill=black, draw=black, shape=circle, inner sep=0pt, minimum width=1.2mm, tikzit category=circuit]
\tikzstyle{cnot targ}=[fill=white, draw=white, shape=circle, tikzit category=circuit, label={center:$\oplus$}, inner sep=0pt, minimum width=2.1mm, tikzit fill={rgb,255: red,102; green,204; blue,255}, tikzit draw=black]
\tikzstyle{ket}=[fill=white, draw=black, shape=regular polygon, regular polygon sides=3, regular polygon rotate=-30, scale=0.7, inner sep=1pt, tikzit category=circuit, tikzit shape=rectangle, tikzit fill=green]
\tikzstyle{bra}=[fill=white, draw=black, shape=regular polygon, regular polygon sides=3, regular polygon rotate=30, scale=0.7, inner sep=1pt, tikzit category=circuit, tikzit shape=rectangle, tikzit fill=red]
\tikzstyle{scalar}=[shape=rectangle, text height=1.5ex, text depth=0.25ex, yshift=0.5mm, fill=white, draw=black, minimum height=5mm, yshift=-0.5mm, minimum width=5mm, font={\small}]
\tikzstyle{clabel}=[fill=white, draw=none, shape=rectangle, tikzit fill={rgb,255: red,56; green,255; blue,242}, font={\footnotesize}, inner sep=1pt, tikzit category=labels]
\tikzstyle{empty diagram}=[draw={gray!40!white}, dashed, shape=rectangle, minimum width=1cm, minimum height=1cm, tikzit category=misc]
\tikzstyle{amap}=[fill=white, draw=black, shape=NEbox, tikzit category=asymmetric, tikzit fill=yellow, tikzit shape=rectangle]
\tikzstyle{amap conj}=[fill=white, draw=black, shape=NWbox, tikzit category=asymmetric, tikzit fill=green, tikzit shape=rectangle]
\tikzstyle{amap adj}=[fill=white, draw=black, shape=SEbox, tikzit category=asymmetric, tikzit fill=red, tikzit shape=rectangle]
\tikzstyle{amap trans}=[fill=white, draw=black, shape=SWbox, tikzit category=asymmetric, tikzit fill=orange, tikzit shape=rectangle]
\tikzstyle{astate}=[fill=white, draw=black, shape=NEtriangle, tikzit category=asymmetric, tikzit shape=circle, tikzit fill=yellow]
\tikzstyle{astate conj}=[fill=white, draw=black, shape=NWtriangle, tikzit category=asymmetric, tikzit shape=circle, tikzit fill=green]
\tikzstyle{astate adj}=[fill=white, draw=black, shape=SEtriangle, tikzit category=asymmetric, tikzit shape=circle, tikzit fill=red]
\tikzstyle{astate trans}=[fill=white, draw=black, shape=SWtriangle, tikzit category=asymmetric, tikzit shape=circle, tikzit fill=orange]

\tikzstyle{blue dash edge}=[-, dashed, dash pattern=on 2pt off 0.5pt, thick, draw={rgb,255: red,68; green,136; blue,255}]
\tikzstyle{green dash edge}=[-, dashed, dash pattern=on 2pt off 0.5pt, thick, draw={rgb,255: red,144; green,238; blue,144}]
\tikzstyle{purple dash edge}=[-, dashed, dash pattern=on 2pt off 0.5pt, thick, draw={rgb,255: red,147; green,112; blue,219}]
\tikzstyle{box edge}=[-, dashed, dash pattern=on 2pt off 0.5pt, thick, draw={rgb,255: red,203; green,192; blue,225}]
\tikzstyle{brace edge}=[-, tikzit draw=blue, decorate, decoration={brace,amplitude=1mm,raise=-1mm}]
\tikzstyle{diredge}=[->]
\tikzstyle{double edge}=[-, double, shorten <=-1mm, shorten >=-1mm, double distance=2pt]
\tikzstyle{gray edge}=[-, {gray!60!white}]
\tikzstyle{pointer edge}=[->, very thick, gray]
\tikzstyle{boldedge}=[-, line width=1.6pt, shorten <=-0.17mm, shorten >=-0.17mm]
\tikzstyle{bidir edge}=[<->, very thick, draw={rgb,255: red,191; green,191; blue,191}]
\tikzstyle{new edge style 0}=[-, fill={rgb,255: red,169; green,94; blue,255}, draw=none, opacity=0.5]
\tikzstyle{new edge style 1}=[->, draw={rgb,255: red,51; green,24; blue,255}, very thick]
\tikzstyle{new edge style 2}=[-, draw=blue]

\definecolor{colorZhZ}{RGB}{255, 255, 255}
\definecolor{colorZhX}{RGB}{200, 200, 200}
\tikzset{
main node/.style={draw, thick, circle, inner sep=0.02cm, minimum size=0.4cm, font=\sffamily\scriptsize}, 
  label node/.style={thick, inner sep=0cm},
    zhZ4/.style n args={4}{
      zx@spider={zhNoPhaseZ}{zhShortZ}{zhLongZ}{stylePhaseInLabelZ}{#1}{#2}{#3}{#4}
    },
    zhX4/.style n args={4}{
      zx@spider={zhNoPhaseX}{zhShortX}{zhLongX}{stylePhaseInLabelX}{#1}{#2}{#3}{#4}
    },
    zhNoPhaseZ/.style={zxNoPhase,fill=colorZhZ},
    zhNoPhaseX/.style={zxNoPhase,fill=colorZhX},
    zhNoPhaseSmallZ/.style={zxNoPhaseSmall,fill=colorZhZ},
    zhNoPhaseSmallX/.style={zxNoPhaseSmall,fill=colorZhX},
    zhShortZ/.style={zxShort,fill=colorZhZ},
    zhShortX/.style={zxShort,fill=colorZhX},
    zhLongZ/.style={zxLong,fill=colorZhZ},
    zhLongX/.style={zxLong,fill=colorZhX},
}
\usepackage{xparse}

\NewDocumentCommand{\phaseone}{ m O{0} O{1} o }{%
  \begin{tikzpicture}[baseline={(0,-0.25)}]
    \begin{scope}[local bounding box=#1]
      \draw (0,0) node (0) [main node] {#2};
      \draw (1,0) node (1) [main node] {#3};
      \IfValueT{#4}{%
        \draw (.5,-.5) node [label node] {$t = #4$};
      }%
      \path (1) edge [in=30, out=60, loop, above right, inner sep=0.02cm] (1);
    \end{scope}
  \end{tikzpicture}%
}

\NewDocumentCommand{\phasezero}{ m O{0} O{1} o }{%
  \begin{tikzpicture}[baseline={(0,-0.25)}]
    \begin{scope}[local bounding box=#1]
      \draw (0,0) node (0) [main node] {#2};
      \draw (1,0) node (1) [main node] {#3};
      \IfValueT{#4}{%
        \draw (.5,-.5) node [label node] {$t = #4$};
      }%
      \path (0) edge [in=30, out=60, loop, above right, inner sep=0.02cm] (0);
    \end{scope}
  \end{tikzpicture}%
}

\NewDocumentCommand{\phaseboth}{ m O{0} O{1} o }{
  \begin{tikzpicture}[baseline={(0,-0.25)}]
    \begin{scope}[local bounding box=#1]
      \draw (0,0) node (0) [main node] {#2};
      \draw (1,0) node (1) [main node] {#3};
      \IfValueT{#4}{
        \draw (.5,-.5) node [label node] {$t = #4$};
      }
      \path (0) edge [in=30, out=60, loop, above right, inner sep=0.02cm] (0);
      \path (1) edge [in=30, out=60, loop, above right, inner sep=0.02cm] (1);
    \end{scope}
  \end{tikzpicture}
}

\NewDocumentCommand{\edgehor}{ m O{0} O{1} o }{
  \begin{tikzpicture}[baseline={(0,-0.25)}]
    \begin{scope}[local bounding box=#1]
      \draw (0,0) node (0) [main node] {#2};
      \draw (1,0) node (1) [main node] {#3};
      \IfValueT{#4}{
        \draw (.5,-.5) node [label node] {$t = #4$};
      }
      \path (0) edge [inner sep=0.02cm] (1);
    \end{scope}
  \end{tikzpicture}
}

\NewExpandableDocumentCommand{\zhZ}{O{}t*t-m}{
  \node[zx main node, zhZ4={#1}{\IfBooleanTF{#3}{-}{}}{\IfBooleanTF{#2}{*}{}}{#4}] {}; 
}
\NewExpandableDocumentCommand{\zhX}{O{}t*t-m}{
  \node[zx main node, zhX4={#1}{\IfBooleanTF{#3}{-}{}}{\IfBooleanTF{#2}{*}{}}{#4}] {}; 
}
\NewExpandableDocumentCommand{\zhH}{O{}t*t-m}{
  \node[zx main node, zhZ4={#1}{\IfBooleanTF{#3}{-}{}}{\IfBooleanTF{#2}{*}{}}{#4}, rectangle] {}; 
}

\newtheorem{observation}{Observation}

\usepackage{mathtools}
\usepackage{ctqw-diagrams}
\usepackage{etoolbox}
\newcounter{appobscounter}
\makeatletter
\newwrite\@obslog
\AtBeginDocument{%
  \setcounter{appobscounter}{0}%
  \gdef\@obskeylist{}%
  \IfFileExists{\jobname.obslog}{\input{\jobname.obslog}}{}%
  \immediate\openout\@obslog=\jobname.obslog\relax
}
\AtEndDocument{\immediate\closeout\@obslog}
\newcommand{\refobs}[1]{%
  \@ifundefined{obsseen@#1}{%
    \expandafter\gdef\csname obsseen@#1\endcsname{1}%
    \immediate\write\@obslog{\string\obsorder{#1}}%
  }{}%
  \ref{obs:gi:#1}%
}
\newcommand{\obsorder}[1]{%
  \@ifundefined{obsnum@#1}{%
    \stepcounter{appobscounter}%
    \expandafter\xdef\csname obsnum@#1\endcsname{\theappobscounter}%
    \xdef\@obskeylist{\@obskeylist,#1}%
  }{}%
}
\newcommand{\obsdef}[3]{%
  \expandafter\long\expandafter\gdef\csname obsbody@#1\endcsname{#3}%
  \expandafter\gdef\csname obstitle@#1\endcsname{#2}%
  \@ifundefined{obsregistered@#1}{%
    \expandafter\gdef\csname obsregistered@#1\endcsname{1}%
    \xdef\@obsregisteredlist{\@obsregisteredlist,#1}%
  }{}%
}
\gdef\@obsregisteredlist{}
\newcommand{\@emitobs}[1]{%
  \ifcsname obsregistered@#1\endcsname
    \@ifundefined{obsemitted@#1}{%
      \expandafter\gdef\csname obsemitted@#1\endcsname{1}%
      \@ifundefined{obsnum@#1}{%
        \stepcounter{appobscounter}%
        \expandafter\xdef\csname obsnum@#1\endcsname{\theappobscounter}%
      }{}%
      \par\vspace{1em}\noindent\rule{\linewidth}{0.4pt}\par
      \subsection*{Observation \csname obsnum@#1\endcsname: \csname obstitle@#1\endcsname}%
      \phantomsection
      \def\@currentlabel{\csname obsnum@#1\endcsname}%
      \hypertarget{obsanchor.app.\csname obsnum@#1\endcsname}{}%
      \def\@currentHref{obsanchor.app.\csname obsnum@#1\endcsname}%
      \label{obs:gi:#1}%
      \csname obsbody@#1\endcsname%
    }{}%
  \fi
}
\newcommand{\printallobs}{%
  \renewcommand{\do}[1]{\@emitobs{##1}}%
  \expandafter\docsvlist\expandafter{\@obskeylist}%
  \renewcommand{\do}[1]{\@ifundefined{obsnum@##1}{\@emitobs{##1}}{}}%
  \expandafter\docsvlist\expandafter{\@obsregisteredlist}%
}

\makeatother

\begin{document}

\preprint{APS/123-QED}

\title{Simplification Rules for Continuous-Time Quantum Walks on Dynamic Graphs}

\author{Mostafa Atallah$^{1,3}$\orcidlink{0009-0004-8187-6932}}
\author{Daniel Dilley$^2$\orcidlink{0000-0002-8821-4059}}
\author{Jishnu Mahmud$^{1}$\orcidlink{0009-0000-3650-8656}}
\author{Zain H. Saleem$^2$\orcidlink{0000-0002-8182-2764}}
\author{Rebekah Herrman$^1$\orcidlink{0000-0001-6944-4206}}\thanks{corresponding author}
\email{rherrma2@utk.edu}
\affiliation{$^1$Department of Industrial and Systems Engineering, University of Tennessee Knoxville, USA}
\affiliation{$^2$Mathematics and Computer Science Division, Argonne National Laboratory, Lemont, IL, USA}
\affiliation{$^3$Department of Physics, Faculty of Science, Cairo University, Giza 12613, Egypt}

\date{\today}

\begin{abstract}
Continuous-time quantum walks (CTQWs) on dynamic graphs realize quantum gates as sequences of time-evolving graph Hamiltonians, but naive constructions produce long sequences with redundancy. Simplification rules, which are graph rewrite rules that shorten a dynamic graph sequence while preserving the unitary it implements, are the CTQW analogue of circuit identities in the gate model. In this work, we give CTQW realizations of the standard single-qubit gates and introduce new graph rewrite rules. We demonstrate the simplification rules through worked circuit reductions and outline their use as transpilation primitives for converting between the circuit model and the dynamic graph framework.
\end{abstract}

\maketitle
\makeatother

\section{Introduction}\label{sec:intro}
A Continuous-Time Quantum Walk (CTQW) on a graph $G$ is a universal model of computation that is the quantum computing analogue to continuous-time random walks in classical computing \cite{childs2009universal, childs2004spatial, herrman2019continuous, farhi1998quantum}. In the CTQW paradigm, the dynamics of a walker on a graph $G$ with adjacency matrix $A_G$ can be described by $e^{-i A_G t}$ where $t \in \mathbb{R}^+$ is the propagation time. While they have a myriad of applications in diverse fields such as computer science, operations research, and biology \cite{chakraborty2020finding, osada2020continuous, apers2022quadratic, tanaka2022spatial, goldsmith2023link, moutinho2023quantum, marsh2020combinatorial, slate2021quantum, d2021protein, santiago2020quantum}, they are in general, difficult to implement in the circuit model of quantum computation. In fact, only CTQWs on some families of graphs can be exactly implemented in the circuit model \cite{qiang2016efficient, portugal2022implementation, adhikari10circuit, loke2017efficient, qu2022deterministic, douglas2009efficient}. Recent work has developed a new graph-based approach for approximating CTQWs in the circuit model of computation \cite{atallah2026simulating} that is inspired by CTQWs on dynamic graphs \cite{herrman2019continuous}. 

CTQWs on dynamic graphs provide a universal framework for quantum computation, where quantum gates are realized as sequences of time-evolving graph-based Hamiltonians~\cite{herrman2019continuous}. A CTQW on a dynamic graph sequence, $\{(G_k, t_k)\}_{k=1}^m$, is a quantum process that is described by $\prod_{k=1}^m e^{-i A_k t_k}$ where a Hamiltonian $A_k$ is the adjacency matrix of graph $G_k$ and $t_k$ is the associated propagation time. In contrast to CTQWs on static graphs, there are clear equivalences between CTQWs on dynamic graphs and gates in the circuit model, as well as between CTQWs on  dynamic graphs and gate-based quantum algorithms such as QAOA \cite{herrman2019continuous, childs2003universal, herrman2022relating, farhi2014quantum, marsh2018quantum}. 

Other research on CTQWs on dynamic graphs has extended the framework in several directions. For example, the authors of \cite{wong2019isolated} showed how adding isolated vertices can reduce the propagation time or the number of edges required when writing some dynamic graph sequences that are equivalent to gates in the circuit model. Other recent research has shown how to write an arbitrary gate in the dynamic graph framework as a sequence of at most three graphs \cite{adisa2021implementing}. However, two primary practical applications of the framework lie in state preparation and Hamiltonian simulation. In \cite{gonzales2025efficient}, the dynamic graph framework is used to motivate an efficient, sparse deterministic state preparation algorithm that uses fewer controlled gates than other sparse state preparation algorithms. While in \cite{atallah2026simulating}, the authors show how to approximate CTQWs by decomposing the underlying graph into matchings, converting each matching to a sequence of gates in the circuit model, and Trotterizing over the matchings. Notably, the algorithm does not require Pauli decomposition, only requires polynomial-time classical overhead, and requires fewer CX gates than Pauli decomposition when used to approximate a CTQW on sparse static graphs with particular structure. These works allude to the fact that CTQWs on dynamic graphs may be useful compiler primitives. 

However, in order for CTQWs on dynamic graphs to be useful compiler primitives, there is a need for more research in the field. In particular, the authors of \cite{herrman2022simplifying} show how one can simplify the dynamic graph sequences if the underlying graphs have specific properties. These simplification techniques are written in terms of graphs, as they can be written in concise pictures and are easy to conceptualize. The authors of the work state that their list of dynamic graph sequence simplifications is not comprehensive. In this work, we introduce dynamic graph sequences that are equivalent to phase gates and $n^{th}$ roots of unitaries in the circuit model, as well as five new dynamic graph sequence simplification techniques. These new gates and simplification techniques allow us to represent more algorithms in the dynamic CTQW framework and write them more compactly. The remainder of this work is organized as follows. In Sec.~\ref{sec:background}, we introduce relevant notation and provide a brief recap of previous dynamic graph simplification techniques. Then we introduce new dynamic graph sequences that are equivalent to a collection of gates in the circuit model of computing in Sec.~\ref{sec:gates}. Next, we introduce new dynamic graph simplification rules in Sec.~\ref{sec:usimplifications}. 
We conclude with a discussion in Sec.~\ref{sec:discussion}.

\section{Background}\label{sec:background}
In this section, we introduce the notation, summarize the CTQW on dynamic graph framework, and list previous CTQW on dynamic graph simplification rules.

\subsection{Notation}
An \textit{undirected graph} $G=(V,E)$ is a collection of vertices $V$ and edges $E$. In general, we will consider undirected graphs that do not contain multiedges, i.e. for any two vertices $u,v$ there exists at most one edge connecting $u$ and $v$; however, we do allow undirected graphs to have \textit{self-loops}, that is, an edge that connects a vertex to itself. In the dynamic CTQW framework, each graph will have $2^n$ vertices where $n \in \mathbb{N}$ is the number of qubits. Each vertex is labeled with a length $n$ bitstring and represents a computational basis state of the underlying Hilbert space. The \textit{adjacency matrix} of an undirected graph $G$, denoted $H_G$, is the symmetric matrix of $0$'s and $1$'s such that if $ij\in E$, then $H_{ij} = H_{ji} = 1$. Note that in the CTQW model of computation, either the adjacency matrix or Laplacian of a graph can be the Hamiltonian. In this work, we use the adjacency matrix formalism and thus denote adjacency matrices by $H$.

Throughout this work, $X_i$ ($Y_i$, $Z_i$) refers to the Pauli $X$- ($Y$-, $Z$-) gate acting on qubit $i$.  Note that we drop the subscript when considering one-qubit systems.

\subsection{Previous simplification rules}
For a given dynamic graph sequence, six simplification rules can be used to produce a new dynamic graph sequence that results in the same Hamiltonian dynamics as the original dynamic graph sequence. This resulting sequence contains either fewer graphs or shorter propagation time(s) of graphs \cite{herrman2022simplifying}.
\begin{observation}\label{obs:commute}
	If sequential graphs commute, their order in the dynamic graph can be swapped.
\end{observation}

\begin{observation}\label{obs:same}
	If $G_{\ell} = G_{\ell + 1}$, then $G_{\ell}$ can be performed for time $t_\ell+t_{\ell+1}$, modulo the period of graph $G_\ell$, and $G_{\ell + 1}$ can be omitted.
\end{observation}

\begin{observation}\label{obs:pst}
	Suppose the dynamic graph has $2^n$ vertices, so the vertices can be labeled in binary. Suppose the dynamic graph contains a sequence of graphs $\{ G_\ell, G_{\ell+1}, \dots, G_m \}$, such that for each graph in this sequence, there is perfect state transfer between every pair of vertices whose binary representations differ in the same locations. Then, the sequence of graphs simply performs a sequence of perfect state transfers, so we can replace $G_\ell$ through $G_m$ with a single graph that performs the resulting perfect state transfers. Any phases can be adjusted using isolated vertices with self-loops.
\end{observation}

\begin{observation}\label{obs:comp}
	Consider sequential graphs $G_\ell$ and $G_{\ell+1}$ with respective adjacency matrices $A_{\ell}$ and $A_{\ell + 1}$. Assume $G_{\ell+1}$ is a subgraph of the complement of $G_\ell$, and furthermore $G_{\ell+1}$ only contains edges or self-loops on vertices that do not have edges or self-loops in $G_\ell$. Then, if $\norm{A_\ell} = \norm{A_{\ell+1}}$, then $G_\ell$ and $G_{\ell+1}$ can be combined into one dynamic graph whose adjacency matrix is $A_\ell + A_{\ell+1}$, and if $t_\ell = t_{\ell+1}$, the evolution time of this graph is $t_\ell$. If $t_\ell \neq t_{\ell+1}$, the graphs can still be combined, but an additional graph will be needed to finish the propagation of the graph with the longer time. This results in a shorter overall time but the same number of graphs in the sequence.
\end{observation}

\begin{observation}\label{obs:singletons}
	If a vertex $v$ is propagated as a looped singleton in $G_\ell$ for time $t_\ell$, and vertex $v$ is not adjacent to any other vertices in $G_j$ for all $a \leq j \leq b$, we can instead propagate vertex $v$ as a looped singleton in any $G_j$ for time $t_j$ between $G_a$ and $G_b$ that satisfies $t_j/\norm{A_j} = t_\ell/\norm{A_\ell}$. If we choose to propagate it as a looped singleton in a graph $G_j$ where it is already connected as a looped singleton, we evolve the singleton by time $2t_j$, modulo the period of $G_j$.
\end{observation}

\begin{observation}\label{obs:hypercube}
	We can use uniform mixing on the hypercube to apply the Hadamard gate to multiple qubits in parallel.
\end{observation}
Importantly, the authors of \cite{herrman2022simplifying} note that there may be other simplification rules. In Sec.~\ref{sec:gates}, we present dynamic graph equivalents of some gates that have not been documented previously, and in Sec.~\ref{sec:usimplifications}, we present new simplification rules.

\section{New gates}\label{sec:gates}
In this section, we introduce dynamic CTQWs that are equivalent to some elementary gates. Note that while some of these observations may seem obvious, inserting a set of gates equivalent to the identity gate can be useful when simplifying CTQW graphs. The phase gates and their conjugate transposes presented here are intuitive consequences of the previously defined $Z, S,$ and $T$ gates, while the Identity and $n-$th root of an arbitrary unitary dynamic graph equivalents require more derivation.

\subsection{Phase gates}\label{sec:phasegates}
We first consider the identity gate, since it serves as a useful building block for later simplifications. Although the identity is also equivalent to an empty dynamic graph sequence and a dynamic graph sequence consisting only of self-loops, each of which is implemented for a time $t = 2k\pi, (k \in \mathbb{Z})$ \cite{wong2019isolated}, it can also be written as a sequence of two graphs, which may aid in reduction in some dynamic graph sequences.

\begin{construct}\label{def:I}
    The identity gate is equivalent to the dynamic graph sequence $\{(A_1, t_1), (A_2,t_2)\}$ where $A_1$ is a two-vertex graph with an edge connecting both vertices, $A_2$ is a two-vertex graph with self-loops on both vertices, and $t_1=t_2=\pi$, as in Fig.~\ref{fig:Identitygate}. 
    Furthermore, each ordered pair  $(A_1, t_1), (A_2,t_2)$ corresponds to an identity gate with a global phase of $-1$.
\end{construct}
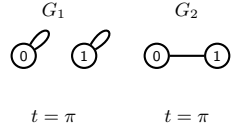
\begin{figure}
\centering
\begin{tikzpicture}[scale=0.8, every node/.style={transform shape}, thick, 
  main node/.style={draw, circle, inner sep=0.02cm, minimum size=0.4cm, font=\sffamily\scriptsize}, 
  label node/.style={inner sep=0cm}]

\begin{scope}[xshift=0cm]
  \draw (0,0) node (1) [main node] {0};
  \draw (1,0) node (2) [main node] {1};
  
  \draw (.5,-1) node [label node] {$t = \pi$};
  \draw (.5,.75) node [label node] {$G_1$};
  
  \path (2) edge [in=30, out=60, loop, above right, inner sep=0.02cm] (2);
   \path (1) edge [in=30, out=60, loop, above right, inner sep=0.02cm] (1);
\end{scope}

\begin{scope}[xshift=2.2cm]
  \draw (0,0) node (1) [main node] {0};
  \draw (1,0) node (2) [main node] {1};
  
  \draw (.5,-1) node [label node] {$t = \pi$};
  \draw (.5,.75) node [label node] {$G_2$};
  
  \draw (1) -- (2);
\end{scope}

\end{tikzpicture}
\caption{Identity gate as a dynamic graph sequence. Note each graph and corresponding time in this sequence corresponds to a $-I$ gate.}\label{fig:Identitygate}
\end{figure}
\begin{proof}
   We first begin with the first construction. It is easy to check that
    \begin{equation*}
        e^{-iA_1 \pi} = \cos{\pi}(\ketbra{0}{0} + \ketbra{1}{1}) -i\sin{\pi}(\ketbra{0}{1} + \ketbra{1}{0}) = -I
    \end{equation*}
    and similarly
        \begin{equation*}
        e^{-iA_2 \pi} = (\cos{\pi}-i\sin{\pi})(\ketbra{0}{0} + \ketbra{1}{1})= -I.
    \end{equation*}
    Furthermore, $(-I)^2 = I$.
\end{proof}
These graphs can also be used to construct the $iI$ and $-iI$ gates.

\begin{construct}\label{def:imaginaryI}
The $iI$ gate is equivalent to the dynamic graph sequence $\{(A_1, \frac{(4k+3)\pi}{2})\}$ where $A_1$ is the adjacency matrix of the graph on two vertices with self-loops on both vertices and the $-iI$ gate is equivalent to the dynamic graph sequence $\{(A_1, \frac{(4k+1)\pi}{2})\}$ . 
\end{construct}
\begin{proof}
    Note that
    \begin{equation*}
        e^{-iA_1 \frac{(4k+3)\pi}{2}} = e^{-iA_1 (2k\pi+\frac{3\pi}{2})}= \cos{\frac{3\pi}{2}}(\ketbra{0}{0} + \ketbra{1}{1}) -i\sin{\frac{3\pi}{2}}(\ketbra{0}{1} + \ketbra{1}{0}) = iI
    \end{equation*}
    and similarly
   \begin{equation*}
        e^{-iA_1 \frac{(4k+1)\pi}{2}} = e^{-iA_1 (2k\pi+\frac{\pi}{2})}= \cos{\frac{\pi}{2}}(\ketbra{0}{0} + \ketbra{1}{1}) -i\sin{\frac{\pi}{2}}(\ketbra{0}{1} + \ketbra{1}{0}) = -iI.
    \end{equation*}
\end{proof}

Next, we construct the dynamic CTQW equivalent of the $P(\theta)$ gate, where
  \begin{align*}
        P(\theta) &= \begin{pmatrix}
            1 & 0 \\
            0 & e^{i \theta} \end{pmatrix}.
    \end{align*}
\begin{construct}
    The $P(\theta)$ gate is equivalent to the dynamic graph sequence $\{(A_1, t_1)\}$ where $A_1$ is a two-vertex graph with a self-loop on the vertex represented by the bitstring $1$ and $t_1=2\pi-\theta$, as in Fig.~\ref{fig:Sgate}.
\end{construct}

\begin{proof}
    Note that
    \begin{equation*}
        e^{-iA_1 (2\pi-\theta)} = e^{-iA_1 (-\theta)}= e^{iA_1 \theta} = \ketbra{0}{0} + e^{i \theta}\ketbra{1}{1}.
    \end{equation*}
\end{proof}

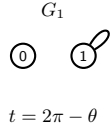
\begin{figure}
\centering
\begin{tikzpicture}[scale=0.8, every node/.style={transform shape}, thick, 
  main node/.style={draw, circle, inner sep=0.02cm, minimum size=0.4cm, font=\sffamily\scriptsize}, 
  label node/.style={inner sep=0cm}]

\begin{scope}[xshift=0cm]
  \draw (0,0) node (1) [main node] {0};
  \draw (1,0) node (2) [main node] {1};
  
  \draw (.5,-1) node [label node] {$t = 2\pi - \theta$};
  \draw (.5,.75) node [label node] {$G_1$};
  
  \path (2) edge [in=30, out=60, loop, above right, inner sep=0.02cm] (2);
\end{scope}

\end{tikzpicture}
\caption{The $P(\theta)$ gate as a dynamic graph sequence.}\label{fig:Sgate}
\end{figure}

\begin{construct}\label{obs:inversephase}
    The dynamic graph sequence that implements $P^\dag (\theta)$ is $\{(A_1, t_1)\}$ where $A_1$ is the graph with two vertices and a single self-loop on the vertex denoted by the bitstring $1$ and $t_1 = \theta$. 
\end{construct}
\begin{proof}
    \begin{equation*}
        e^{-iA_1 (\theta)} =  \ketbra{0}{0} + e^{-i \theta}\ketbra{1}{1}.
    \end{equation*}
\end{proof}

Note that we can immediately write $S^{\dagger}$ and $T^{\dagger}$ with this construction.

\begin{construct}\label{obs:inverseS}
    The dynamic graph sequence that implements $S^\dagger$ is $\{(A_1, t_1)\}$ where $A_1$ is the graph with a single self-loop on the vertex denoted by the bitstring $1$ and $t_1 = \pi/2$.  
\end{construct}

\begin{construct}\label{obs:inverseT}
    The dynamic graph sequence that implements $T^\dagger$ is $\{(A_1, t_1)\}$ where $A_1$ is the graph with a single self-loop on the vertex denoted by the bitstring $1$ and $t_1 = \pi/4$.  
\end{construct}

All generalized phase gates are summarized in Table~\ref{tab:phase_gates}.

\begin{table}[ht]
\centering
\renewcommand{\arraystretch}{3.5}
\begin{tabular}{c|c|c|l}
\textbf{Gate} & \textbf{Matrix} & \textbf{CTQW Graph} & \textbf{Source} \\ \hline
$Z$ & $\begin{psmallmatrix} 1 & 0 \\ 0 & -1 \end{psmallmatrix}$ & $\Gloop{t=\pi}$ & Wong (2019) \cite{wong2019isolated}
\\
$S$ & $\begin{psmallmatrix} 1 & 0 \\ 0 & i \end{psmallmatrix}$ & $\Gloop{t=\frac{3\pi}{2}}$ & Wong (2019) \cite{wong2019isolated}
\\
$T$ & $\begin{psmallmatrix} 1 & 0 \\ 0 & e^{i\pi/4} \end{psmallmatrix}$ & $\Gloop{t=\frac{7\pi}{4}}$ & Wong (2019) \cite{wong2019isolated}
\\
\hline
$P(\lambda)$ & $\begin{psmallmatrix} 1 & 0 \\ 0 & e^{i\lambda} \end{psmallmatrix}$ & $\Gloop{t=2\pi{-}\lambda}$ & Wong (2019) \cite{wong2019isolated}; Gonzales et al.~(2025) \cite{gonzales2025efficient} \\
$P^{\dagger}(\lambda)$ & $\begin{psmallmatrix} 1 & 0 \\ 0 & e^{-i\lambda} \end{psmallmatrix}$ & $\Gloop{t=\lambda}$ & Derived in~\ref{obs:inversephase}\\ 
$S^\dagger$ & $\begin{psmallmatrix} 1 & 0 \\ 0 & -i \end{psmallmatrix}$ & $\Gloop{t=\frac{\pi}{2}}$ & Derived in~\ref{obs:inverseS} \\
$T^\dagger$ & $\begin{psmallmatrix} 1 & 0 \\ 0 & e^{-i\pi/4} \end{psmallmatrix}$ & $\Gloop{t=\frac{\pi}{4}}$ & Derived in~\ref{obs:inverseT} \\
\end{tabular}
\caption{Generalized phase gates: matrix definition, CTQW graph, and source. All entries are realized by a single self-loop graph at an appropriate time.}
\label{tab:phase_gates}
\end{table}
\subsection{Rotation gates}\label{sec:rotationgates}

In this subsection, we present CTQW implementations of the continuously parameterized rotation gates $R_y(\alpha)$ and $R_z(\alpha)$. Note that $R_x(\alpha)$ has been defined in \cite{gonzales2025efficient}. $R_z$ requires two graphs (one self-loop on each vertex) for an exact, global-phase-free implementation while $R_y$ uses a three-graph $S R_x(\alpha) S^\dagger$ sequence. All three rotation gates $R_x(\alpha)$, $R_y(\alpha)$, and $R_z(\alpha)$ together generate $SU(2)$ via the ZYZ decomposition. While general single qubit unitaries can be implemented using the construction from \cite{adisa2021implementing}, we state the below rotation gates explicitly, as the ZYZ decomposition will be useful for deriving the $n-$th root of a unitary later in this section. Before deriving the equivalent CTQW on dynamic graph sequences, recall that

\begin{equation*}
    R_y(\alpha) \;=\; e^{-i\alpha Y/2} \;=\; \begin{pmatrix} \cos(\alpha/2) & -\sin(\alpha/2) \\ \sin(\alpha/2) & \cos(\alpha/2) \end{pmatrix}
\end{equation*}
and 

\[
R_z(\alpha) \;=\; e^{-i\alpha Z/2} \;=\; \begin{pmatrix} e^{-i\alpha/2} & 0 \\ 0 & e^{i\alpha/2} \end{pmatrix}.
\]

\begin{construct}
The $R_y(\alpha)$ gate can be decomposed as $R_y(\alpha) = R_z(\frac{\pi}{2})\, R_x(\alpha)\, R_z(-\frac{\pi}{2}) = S \cdot R_x(\alpha) \cdot S^\dagger$ since $S = e^{i \frac{\pi}{4}} R_z(\frac{\pi}{2})$ and $S^\dagger = e^{-i \frac{\pi}{4}} R_z(-\frac{\pi}{2})$. This is equivalent to the dynamic graph sequence $\{(A_1, t_1), (A_2,t_2), (A_3,t_3)\}$ where $A_1 = A_3$ is a two-vertex graph with a self-loop on $\ket{1}$, $A_2$ is a two-vertex graph with an edge connecting both vertices, $t_1 = \pi/2$, $t_2 = \alpha/2$, and $t_3 = 3\pi/2$.
\end{construct}

\begin{proof}
Exponentiating each adjacency matrix yields
\begin{align*}
 & e^{-iA_1 t_1} \;=\; \ketbra{0}{0} - i\,\ketbra{1}{1}, \\
  &e^{-iA_2 t_2} \;=\; \cos\frac{\alpha}{2}(\ketbra{0}{0} + \ketbra{1}{1}) - i\sin\frac{\alpha}{2}(\ketbra{0}{1} + \ketbra{1}{0}) \\
  & e^{-iA_3 t_3} \;=\; \ketbra{0}{0} + i\,\ketbra{1}{1}.
\end{align*}

Multiplying them together yields the result.
\end{proof}

Next, we formalize the dynamic CTQW equivalent to the $R_z(\alpha)$ gate.

\begin{construct}
The $R_z(\alpha)$ is equivalent to the dynamic graph sequence $\{(A_1, t_1), (A_2,t_2)\}$ where $A_1$ is a two-vertex graph with a self-loop on $\ket{0}$, $A_2$ is a two-vertex graph with a self-loop on $\ket{1}$, $t_1 = \alpha/2$, and $t_2 = 2\pi- \alpha/2$.
\end{construct}
\begin{proof}
Exponentiating each adjacency matrix yields
\begin{align*}
& e^{-iA_1 t_1} \;=\; \bigl[\cos(\alpha/2) - i\sin(\alpha/2)\bigr]\,\ketbra{0}{0} + \ketbra{1}{1}, \\[2pt]
& e^{-iA_2 t_2} \;=\; \ketbra{0}{0} + e^{i\alpha/2}\,\ketbra{1}{1}.
\end{align*}
Multiplying them together yields the result.
\end{proof}
All rotation gate implementations are summarized in Table~\ref{tab:rotation_gate_summary}.

\begin{table}[ht]
\centering
\renewcommand{\arraystretch}{3.5}
\begin{tabular}{c|c|l}
\textbf{Gate} & \textbf{CTQW Graph Sequence} & \textbf{Source} \\ \hline
$R_x(\alpha)$ & $\Gedge{t=\frac{\alpha}{2}}$ &  Gonzales et al.~(2025) \cite{gonzales2025efficient} \\
$R_z(\alpha)$ & $\Gloopz{t=\frac{\alpha}{2}} \;\; \Gloop{t=2\pi{-}\frac{\alpha}{2}}$ & Derived \\
$R_y(\alpha)$ & $\Gloop{t=\frac{\pi}{2}} \;\; \Gedge{t=\frac{\alpha}{2}} \;\; \Gloop{t=\frac{3\pi}{2}}$ & Derived \\
$U_{(\theta,\phi,\lambda)}$ & $\Gloop{t=\frac{7\pi}{2}-\phi} \;\; \Gedge{t=\frac{\theta}{2}} \;\; \Gloop{t=\frac{5\pi}{2}-\lambda}$ & Adisa \& Wong (2021) \cite{adisa2021implementing} \\
\end{tabular}
\caption{Continuously parameterized rotation gate implementations. $R_x$ uses a single edge; $R_z$ requires two graphs (one self-loop on each vertex) for an exact, global-phase-free implementation; $R_y$ uses three graphs via the $S R_x S^\dagger$ decomposition.}
\label{tab:rotation_gate_summary}
\end{table}

\subsection{$n$-th roots of unitaries}\label{sec:othergates}

In \cite{adisa2021implementing}, Adisa and Wong showed that any single-qubit gate can be implemented as a length-3 dynamic CTQW by connecting the ZYZ Euler decomposition of $SU(2)$ to the two-vertex graph primitives. In this section we use their construction to derive the principal $n$-th root $U^{1/n}$ of an arbitrary single-qubit gate via De Moivre's theorem. The ZYZ decomposition is the standard Euler angle parametrization of $SU(2)$ (see, e.g., \cite{nielsen2011quantumCompAndQuantInfo}). Any single-qubit unitary can be written as (up to global phase)
\begin{equation}
U_{(\theta,\phi,\lambda)} \;=\; R_z(\phi)\, R_y(\theta)\, R_z(\lambda) \;=\; \begin{pmatrix}
\cos(\theta/2) & -e^{i\lambda}\sin(\theta/2) \\
e^{i\phi}\sin(\theta/2) & e^{i(\phi+\lambda)}\cos(\theta/2)
\end{pmatrix}
\label{eq:ZYZ}
\end{equation}

\begin{thm}[Length-3 dynamic walk, \cite{adisa2021implementing}]\label{thm:length3}
Any single-qubit gate $U_{(\theta,\phi,\lambda)}$ can be implemented up to global phase by a length-3 dynamic CTQW
\begin{equation}
U_{(\theta,\phi,\lambda)} \;=\;\;
\Gloop{\frac{7\pi}{2} - \phi} \;\;
\Gedge{\frac{\theta}{2}} \;\;
\Gloop{\frac{5\pi}{2} - \lambda}
\label{eq:gate_U}
\end{equation}
where all times are taken mod $2\pi$.
\end{thm}

Given this construction, it is natural to consider how the CTQW for the principal $n$-th root $U^{1/n}$ relates to the CTQW for $U$. The answer is not the naive parameter scaling $(\theta/n, \phi/n, \lambda/n)$, but rather the correct parameters can be derived from De Moivre's theorem applied to the spectral decomposition of $U$.

Any single-qubit unitary $U \in U(2)$, which is the group of $2 \times 2$ complex matrices $U$ satisfying $U U^\dag = I$, can be written in axis-angle form as
\begin{equation}
U \;=\; e^{i\alpha}\,e^{-i(\hat{a}\cdot\vec\sigma)\frac{\beta}{2}}
\label{eq:axisangle}
\end{equation}
where $\alpha \in [0, 2\pi)$ is the global phase, $\beta \in [0, 2\pi)$ is the rotation angle, $\hat{a} \in S^2$ is the rotation axis (a unit Bloch vector), and $\vec\sigma = (X, Y, Z)$. The principal $n$-th root acts on the same eigenspaces, sending each eigenvalue to its principal $n$-th root,
\begin{equation}
U^{1/n} \;=\; e^{i\frac{\alpha}{n}}\,e^{-i(\hat{a}\cdot\vec\sigma)\frac{\beta}{2n}}.
\label{eq:axisanglenth}
\end{equation}
The axis $\hat{a}$ is preserved, the rotation angle is scaled to $\beta/n$, and the global phase is scaled to $\alpha/n$.

\begin{construct}
    $U^{1/n}$ is equivalent to the dynamic graph sequence $\{(A_1, t_1), (A_2, t_2), (A_3, t_3)\}$ with the same $A_i$ and $t_i$ as in Theorem~\ref{thm:length3} and
    \begin{align*}
\theta_n \;&=\; 2\arccos\sqrt{\cos^2\!\bigl(\tfrac{\beta}{2n}\bigr) + a_z^2\sin^2\!\bigl(\tfrac{\beta}{2n}\bigr)}, \\
\phi_n \;&=\; \arctan\!\bigl(a_z\tan\!\bigl(\tfrac{\beta}{2n}\bigr)\bigr) + \arctan\!\bigl(-\tfrac{a_x}{a_y}\bigr), \\
\lambda_n \;&=\; \arctan\!\bigl(a_z\tan\!\bigl(\tfrac{\beta}{2n}\bigr)\bigr) - \arctan\!\bigl(-\tfrac{a_x}{a_y}\bigr).
\end{align*}
\end{construct}

\begin{proof}
Given $U_{(\theta,\phi,\lambda)}$, we build $U^{1/n}$ in two steps. First, extract axis-angle $(\hat{a}, \beta, \alpha)$ from $(\theta, \phi, \lambda)$ via the standard $SU(2)$ formulas, after factoring out the global phase $\alpha = (\phi+\lambda)/2$
\begin{align}
\cos\!\bigl(\tfrac{\beta}{2}\bigr) \;&=\; \cos\!\bigl(\tfrac{\theta}{2}\bigr)\cos\!\bigl(\tfrac{\phi+\lambda}{2}\bigr), \nonumber \\
a_x\sin\!\bigl(\tfrac{\beta}{2}\bigr) \;&=\; -\sin\!\bigl(\tfrac{\theta}{2}\bigr)\sin\!\bigl(\tfrac{\phi-\lambda}{2}\bigr), \label{eq:zyz2axis} \\
a_y\sin\!\bigl(\tfrac{\beta}{2}\bigr) \;&=\; \sin\!\bigl(\tfrac{\theta}{2}\bigr)\cos\!\bigl(\tfrac{\phi-\lambda}{2}\bigr), \nonumber \\
a_z\sin\!\bigl(\tfrac{\beta}{2}\bigr) \;&=\; \cos\!\bigl(\tfrac{\theta}{2}\bigr)\sin\!\bigl(\tfrac{\phi+\lambda}{2}\bigr). \nonumber
\end{align}
Second, $U^{1/n}$ is realized by the same length-3 CTQW (Eq.~\eqref{eq:gate_U}) with ZYZ Euler angles $(\theta_n, \phi_n, \lambda_n)$ obtained by replacing $\beta$ with $\beta/n$ in the inverse map of Eq.~\eqref{eq:zyz2axis} while keeping the axis $\hat{a}$ fixed:
\begin{align}
\theta_n \;&=\; 2\arccos\sqrt{\cos^2\!\bigl(\tfrac{\beta}{2n}\bigr) + a_z^2\sin^2\!\bigl(\tfrac{\beta}{2n}\bigr)}, \nonumber \\
\phi_n \;&=\; \arctan\!\bigl(a_z\tan\!\bigl(\tfrac{\beta}{2n}\bigr)\bigr) + \arctan\!\bigl(-\tfrac{a_x}{a_y}\bigr), \label{eq:nthroot_angles} \\
\lambda_n \;&=\; \arctan\!\bigl(a_z\tan\!\bigl(\tfrac{\beta}{2n}\bigr)\bigr) - \arctan\!\bigl(-\tfrac{a_x}{a_y}\bigr). \nonumber
\end{align}
Substituting into Eq.~\eqref{eq:gate_U},
\begin{equation}
U^{1/n} \;=\; \Gloop{\frac{7\pi}{2} - \phi_n}\,\Gedge{\frac{\theta_n}{2}}\,\Gloop{\frac{5\pi}{2} - \lambda_n}.
\label{eq:gate_Unth}
\end{equation}
Only $\beta/n$ inside the $\arctan$ and $\arccos$ depends on $n$, while the axis-fixed term $\arctan(-a_x/a_y)$ in $\phi_n$ and $\lambda_n$ is independent of $n$. The global phase $e^{i\alpha/n}$ is absorbed by Eq.~\eqref{eq:gate_U}, which already realizes $U$ up to global phase. At $n = 1$, $(\theta_1, \phi_1, \lambda_1)$ reduce to the input ZYZ angles for $U$ itself.
\end{proof}

When $\hat{a}$ aligns with a coordinate axis, that is when $U$ is a Pauli-axis gate up to global phase, the recipe trivializes to direct angle scaling. For any involution $A$ (satisfying $A^2 = I$), the propagator $e^{-iAs}$ has spectrum $\{e^{-is}, e^{is}\}$ on the $\pm 1$ eigenspaces of $A$, and De Moivre directly gives
\[
\bigl(e^{-iAs}\bigr)^{1/n} \;=\; e^{-iA\frac{s}{n}} \;=\; \cos\!\bigl(\tfrac{s}{n}\bigr)\,I - i\sin\!\bigl(\tfrac{s}{n}\bigr)\,A.
\]
So taking the $n$-th root of an involution-exponential simply divides the time by $n$. This recovers the unified family of $n$-th roots of Pauli-axis gates collected in Table~\ref{tab:nth_roots}, where the $Z^{1/n}$ row reproduces the standard diagonal phase gates ($S = Z^{1/2}$, $T = Z^{1/4}$, etc.) from the above single-loop realizations.

\begin{table}[h]
\centering
\renewcommand{\arraystretch}{1.5}
\begin{tabular}{c|c}
\textbf{Gate} & \textbf{$n$-th root CTQW} \\ \hline
$X^{1/n}$ & $\Gboth{2\pi - \frac{\pi}{2n}}\,\Gedge{\frac{\pi}{2n}}$ \\
$Z^{1/n}$ & $\Gboth{2\pi - \frac{\pi}{2n}}\,\Gloop{\frac{\pi}{n}}$ \\
$Y^{1/n}$ & $\Gloop{\frac{\pi}{2}}\,\Gedge{\frac{\pi}{2n}}\,\Gloop{\frac{3\pi}{2}}$ \\
\end{tabular}
\caption{$n$-th roots of Pauli-axis gates as dynamic graph sequences.}
\label{tab:nth_roots}
\end{table}

\begin{example}
As a worked example, consider $\sqrt{H}$. The Hadamard $H = \frac{1}{\sqrt 2}(X + Z)$ is a rotation by $\pi$ about the axis $\hat{a} = (1/\sqrt 2, 0, 1/\sqrt 2)$, with global phase $\alpha = \pi/2$, so $H = i\,e^{-i(\hat{a}\cdot\vec\sigma)\frac{\pi}{2}}$. Substituting into the axis-angle form (Eq.~\eqref{eq:axisanglenth}) at $n = 2$,
\begin{align*}
\sqrt{H} \;&=\; e^{i\frac{\pi}{4}}\,e^{-i(\hat{a}\cdot\vec\sigma)\frac{\pi}{4}} \\
&=\; e^{i\frac{\pi}{4}}\,\bigl(\cos\!\bigl(\tfrac{\pi}{4}\bigr)\,I - i\sin\!\bigl(\tfrac{\pi}{4}\bigr)\,\hat{a}\cdot\vec\sigma\bigr) \\
&=\; \frac{1+i}{2}\,I + \frac{1-i}{2\sqrt 2}\,(X + Z),
\end{align*}
and one can verify by Pauli arithmetic that $(\sqrt H)^2 = H$. To realize $\sqrt H$ as a length-3 CTQW, run the recipe from Eq.~\eqref{eq:zyz2axis}. For $H = R_z(0)R_y(\pi/2)R_z(\pi)$ (up to global phase), $(\theta, \phi, \lambda) = (\pi/2, 0, \pi)$, so $\alpha = \pi/2$, $\cos(\beta/2) = 0$ giving $\beta = \pi$, and the axis components are $a_x = 1/\sqrt 2$, $a_y = 0$, $a_z = 1/\sqrt 2$. At $n = 2$, the new ZYZ angles are $\theta_2 = \pi/3$, $\phi_2 = \arctan(1/\sqrt 2) - \pi/2$, and $\lambda_2 = \arctan(1/\sqrt 2) + \pi/2$. Substituting into Eq.~\eqref{eq:gate_Unth} and reducing the boundary-loop times modulo $2\pi$,
\begin{equation}
\sqrt H \;=\; \Gloop{2\pi - \eta} \;\; \Gedge{\pi/6} \;\; \Gloop{2\pi - \eta},
\label{eq:sqrtH_chain}
\end{equation}
where $\eta := \arctan(1/\sqrt{2})$. Both boundary loops have the same time after period reduction (length-3 form up to a global phase, as in Theorem~\ref{thm:length3}). 

The time-ordered $\sqrt H$ propagator (leftmost-first means rightmost in matrix order) is
\begin{align}
U_{\sqrt H} \;=\; G_L\,G_E\,G_L
\;&=\; \begin{pmatrix} 1 & 0 \\ 0 & e^{i\arctan(1/\sqrt{2})} \end{pmatrix}\begin{pmatrix} \tfrac{\sqrt 3}{2} & -\tfrac{i}{2} \\ -\tfrac{i}{2} & \tfrac{\sqrt 3}{2} \end{pmatrix}\begin{pmatrix} 1 & 0 \\ 0 & e^{i\arctan(1/\sqrt{2})} \end{pmatrix} \\[2pt]
\;&=\; \begin{pmatrix} \tfrac{\sqrt 3}{2} & -\tfrac{i}{2}\,e^{i\eta} \\ -\tfrac{i}{2}\,e^{i\eta} & \tfrac{\sqrt 3}{2}\,e^{i 2\eta} \end{pmatrix}, \label{eq:sqrtH}
\end{align}
where we have set $\eta := \arctan(1/\sqrt 2)$, $\cos\eta = \sqrt{2/3}$, $\sin\eta = 1/\sqrt 3$, so $e^{i\eta} = \tfrac{\sqrt 6 + i \sqrt 3}{3}$ and $e^{i 2\eta} = \tfrac{1 + 2\sqrt 2\,i}{3}$. Substituting these values entry by entry,
\begin{equation*}
U_{\sqrt H} \;=\; \frac{1}{6}\begin{pmatrix} 3\sqrt 3 & \sqrt 3 - i\sqrt 6 \\ \sqrt 3 - i\sqrt 6 & \sqrt 3 + i\,2\sqrt 6 \end{pmatrix}.
\end{equation*}
This is the closed-form $\sqrt H$ matrix multiplied by the residual global phase $e^{i(\eta - \pi/4)}$ of the length-$3$ Adisa--Wong construction. Directly,
\begin{equation*}
e^{i(\eta - \pi/4)}\,\sqrt H \;=\; e^{i(\eta - \pi/4)} \biggl(\frac{1+i}{2}\,I + \frac{1-i}{2\sqrt 2}(X + Z)\biggr) \;=\; \frac{1}{6}\begin{pmatrix} 3\sqrt 3 & \sqrt 3 - i\sqrt 6 \\ \sqrt 3 - i\sqrt 6 & \sqrt 3 + i\,2\sqrt 6 \end{pmatrix},
\end{equation*}
which matches $U_{\sqrt H}$ above, using $e^{i(\eta - \pi/4)} = e^{i\eta}\cdot e^{-i\pi/4} = \tfrac{\sqrt 6 + i\sqrt 3}{3}\cdot\tfrac{1 - i}{\sqrt 2}$. Equivalently, $U_{\sqrt H} = e^{i(\eta - \pi/4)}\,\sqrt H$, where $\sqrt H$ itself is the gate $\tfrac{1+i}{2}\,I + \tfrac{1-i}{2\sqrt 2}(X+Z)$ derived in the worked example above.

\textit{Squaring to verify $(\sqrt H)^2 = H$.} Squaring Eq.~\eqref{eq:sqrtH} yields
\begin{align*}
U_{\sqrt H}^{\,2}
\;&=\; \begin{pmatrix} \tfrac{\sqrt 3}{2} & -\tfrac{i}{2}\,e^{i\eta} \\ -\tfrac{i}{2}\,e^{i\eta} & \tfrac{\sqrt 3}{2}\,e^{i 2\eta} \end{pmatrix}^{\!2}
\\[2pt]
\;&=\; \begin{pmatrix} \tfrac{3}{4} - \tfrac{1}{4}\,e^{i 2\eta} & -\tfrac{i\sqrt 3}{4}\bigl(1 + e^{i 2\eta}\bigr)\,e^{i\eta} \\ -\tfrac{i\sqrt 3}{4}\bigl(1 + e^{i 2\eta}\bigr)\,e^{i\eta} & \tfrac{3}{4}\,e^{i 4\eta} - \tfrac{1}{4}\,e^{i 2\eta} \end{pmatrix}.
\end{align*}
Using $e^{i 2\eta} = \tfrac{1}{3} + \tfrac{2\sqrt 2}{3}\,i$ and the identity $1 + e^{i 2\eta} = 2\cos\eta\,e^{i\eta}$ (since $|1 + e^{i 2\eta}| = 2\cos\eta = 2\sqrt{2/3}$ and $\arg(1+e^{i 2\eta}) = \eta$), the off-diagonal entries simplify to
\begin{equation*}
-\tfrac{i\sqrt 3}{4}\cdot 2\cos\eta\,e^{i\eta} \cdot e^{i\eta} \;=\; -\tfrac{i\sqrt 3\cos\eta}{2}\,e^{i 2\eta} \;=\; -\tfrac{i\sqrt 2}{2}\,e^{i 2\eta},
\end{equation*}
where the last step uses $\sqrt 3\cos\eta = \sqrt 3 \cdot \sqrt{2/3} = \sqrt 2$. For the diagonal entries, the closed forms of $e^{i 2\eta}$ and $e^{i 4\eta} = \cos 4\eta + i\sin 4\eta = -\tfrac{7}{9} + \tfrac{4\sqrt 2}{9}\,i$ give
\begin{equation*}
\tfrac{3}{4} - \tfrac{1}{4}\,e^{i 2\eta} \;=\; \tfrac{2}{3} - \tfrac{i\sqrt 2}{6}, \qquad \tfrac{3}{4}\,e^{i 4\eta} - \tfrac{1}{4}\,e^{i 2\eta} \;=\; -\tfrac{2}{3} + \tfrac{i\sqrt 2}{6}.
\end{equation*}
Factoring out $e^{i(2\eta - \pi/2)} = \sin 2\eta - i\cos 2\eta = \tfrac{2\sqrt 2 - i}{3}$ from every entry, we obtain
\begin{equation*}
U_{\sqrt H}^{\,2} \;=\; \frac{2\sqrt 2 - i}{3}\,\begin{pmatrix} \tfrac{1}{\sqrt 2} & \tfrac{1}{\sqrt 2} \\ \tfrac{1}{\sqrt 2} & -\tfrac{1}{\sqrt 2} \end{pmatrix} \;=\; e^{i(2\eta - \pi/2)}\,H.
\end{equation*}
The residual global phase $e^{i(2\eta - \pi/2)} = \bigl(e^{i(\eta - \pi/4)}\bigr)^2$ is exactly the square of the single-copy phase identified above, consistent with $U_{\sqrt H}^{\,2} = \bigl(e^{i(\eta - \pi/4)}\,\sqrt H\bigr)^2 = e^{i(2\eta - \pi/2)}\,(\sqrt H)^2 = e^{i(2\eta - \pi/2)}\,H$, which verifies $(\sqrt H)^2 = H$ at the CTQW level (with the global phase tracked exactly).
\end{example}
\begin{figure}
\centering
\begin{tikzpicture}[scale=0.8, every node/.style={transform shape}, thick, 
  main node/.style={draw, circle, inner sep=0.02cm, minimum size=0.4cm, font=\sffamily\scriptsize}, 
  label node/.style={inner sep=0cm}]

\begin{scope}[xshift=0cm]
  \draw (0,0) node (1) [main node] {0};
  \draw (1,0) node (2) [main node] {1};
  
  \draw (.5,-1) node [label node] {$t = 3\pi/2$};
  \draw (.5,.75) node [label node] {$G_1$};
  
  \path (2) edge [in=30, out=60, loop, above right, inner sep=0.02cm] (2);
   \path (1) edge [in=30, out=60, loop, above right, inner sep=0.02cm] (1);
\end{scope}

\begin{scope}[xshift=2.2cm]
  \draw (0,0) node (1) [main node] {0};
  \draw (1,0) node (2) [main node] {1};
  
  \draw (.5,-1) node [label node] {$t = \pi/2$};
  \draw (.5,.75) node [label node] {$G_2$};
  
  \draw (1) -- (2);
\end{scope}

\end{tikzpicture}
\caption{Original $X$ gate}\label{fig:Xgate}
\end{figure}
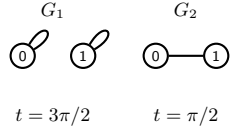

\section{Graph simplification rules}\label{sec:usimplifications} 

As motivation, recall that $HZH=X$. Note that Fig.~\ref{fig:HZH} simplifies to Fig.~\ref{fig:newX} by the simplification rules from \cite{herrman2022simplifying}. We will show that Fig.~\ref{fig:newX} is equivalent to the $X$ gate. First, the adjacency matrices for $G_1$ and $G_2$ in Fig.~\ref{fig:HZH} respectively are $A_1$ and $A_2$ (note $G_1 = G_3$ so $A_1 = A_3$).
\begin{equation*}
  A_1 = 
  \begin{pmatrix}
    0 & 0 \\
    0 & 1
  \end{pmatrix} \; \; \; \; \;\;\;\;\;\;\;
  A_2 = 
  \begin{pmatrix}
    0 & 1 \\
    1 & 0 
  \end{pmatrix}.
\end{equation*}
Note that 
\begin{align*}
e^{-i A_1 3\pi/2}&e^{-i A_2 \pi/2}e^{-i A_1 3\pi/2}= \\
 & \begin{pmatrix}
    1 & 0 \\
    0 & \cos(3\pi/2) -i \sin(3\pi/2)
  \end{pmatrix} 
  \begin{pmatrix}
    0 & \cos(\pi/2) - i \sin(\pi/2) \\
    \cos(\pi/2) - i \sin(\pi/2) & 0 
  \end{pmatrix}
    \begin{pmatrix}
    1 & 0 \\
    0 & \cos(3\pi/2) -i \sin(3\pi/2)
  \end{pmatrix} = \\
&    \begin{pmatrix}
    1 & 0 \\
    0 & i
  \end{pmatrix} 
  \begin{pmatrix}
    0 & -i \\
    -i & 0 
  \end{pmatrix}
    \begin{pmatrix}
    1 & 0 \\
    0 & i
  \end{pmatrix} = \\
 &  \begin{pmatrix}
    0 & 1 \\
    1 & 0
  \end{pmatrix} ,
\end{align*}
which is the X-gate. This observation inspires our first dynamic graph simplification rule.
\begin{obs}\label{obs:movingselfloop}
Let $A_k$ be the adjacency matrix of a graph that has a single self loop on vertex $k$, $A_j$ be the adjacency matrix of a graph with a single self-loop on vertex $j$, and $A_{jk}$ be the adjacency matrix of a graph that has an edge connecting vertices $j$ and $k$. If $t_{jk} = \frac{(2m+1) \pi}{2}$ and $t_j = t_k$, then $e^{-i A_k t_k}e^{-i A_{jk} t_{jk}}=e^{-i A_{jk} t_{jk}}e^{-i A_j t_j}$.
\end{obs}

\begin{proof}
    Note that
    \begin{equation*}
        \exp{-i A_j t_j} = (\cos(t_j) - i \sin(t_j))\ketbra{j}{j} + \sum_{\ell \neq j} \ketbra{\ell}{\ell}
    \end{equation*}
    and
        \begin{equation*}
        \exp{-i A_{jk} t_{jk}} = (\cos(t_{jk}))(\ketbra{j}{j} + \ketbra{k}{k}) -i \sin(t_{jk})(\ketbra{k}{j}+\ketbra{j}{k}) + \sum_{\ell \neq j,k} \ketbra{\ell}{\ell}.
    \end{equation*}

  When $t_{jk} = \frac{(2m+1)\pi}{2}$ and $t_j = t_k$,
    \begin{align*}
        e^{-i A_k t_k}e^{-i A_{jk} t_{jk}} & = [(\cos(t_k) - i \sin(t_k))\ketbra{k}{k} + \sum_{\ell \neq j,k} \ketbra{\ell}{\ell}][ -1(\ketbra{j}{j} + \ketbra{k}{k}) + (-1)^m i(\ketbra{k}{j}+\ketbra{j}{k})  + \sum_{\ell \neq j, k} \ketbra{\ell}{\ell}] \\
    \end{align*}

        The expression in the right bracket is a matrix that has ``1" on all diagonal entries except at positions $j$ and $k$, which have ``0", and entries $jk$ and $kj$ have $(-1)^m i$. The expression in the left bracket is a diagonal matrix with all diagonal entries ``1" except $k$, which has entry ``$e^{-i t_k}$". Multiplying these together yields the matrix with all ``1" entries on the diagonal except position $j$ has a ``$i$" and position $k$ has a ``$ie^{-i t_k}$".

    Furthermore,
    \begin{align*}
    e^{-i A_{jk} t_{jk}}e^{-i A_j t_j} & = [\cos(t_{jk})(\ketbra{j}{j} + \ketbra{k}{k}) -i \sin(t_{jk})(\ketbra{k}{j}+\ketbra{j}{k})  + \sum_{\ell \neq j,k} \ketbra{\ell}{\ell}][(\cos(t_j) - i \sin(t_j))A_j + \sum_{\ell \neq j} \ketbra{\ell}{\ell}] \\
    &= [-1(\ketbra{j}{j} + \ketbra{k}{k}) + (-1)^m i (\ketbra{k}{j}+\ketbra{j}{k})  + \sum_{\ell \neq j,k} \ketbra{\ell}{\ell}][(\cos(t_k) - i \sin(t_k))\ketbra{j}{j} + \sum_{\ell \neq j,k} \ketbra{\ell}{\ell}].
    \end{align*}
    The expression in the left bracket here is the same as the right bracket above. The expression in the right bracket is a diagonal matrix with all diagonal entries ``1" except $j$, which has entry ``$e^{-i t_k}$". Multiplying these together yields the matrix with all ``1" entries on the diagonal except position $j$ has a ``$i$" and position $k$ has a ``$ie^{-i t_k}$".
\end{proof}

\begin{figure}
\centering
\begin{tikzpicture}[scale=0.8, every node/.style={transform shape}, thick, 
  main node/.style={draw, circle, inner sep=0.02cm, minimum size=0.4cm, font=\sffamily\scriptsize}, 
  label node/.style={inner sep=0cm}]

\begin{scope}[xshift=0cm]
  \draw (0,0) node (1) [main node] {0};
  \draw (1,0) node (2) [main node] {1};
  
  \draw (.5,-1) node [label node] {$t = 3\pi/2$};
  \draw (.5,.75) node [label node] {$G_1$};
  
  \path (2) edge [in=30, out=60, loop, above right, inner sep=0.02cm] (2);
\end{scope}

\begin{scope}[xshift=2.2cm]
  \draw (0,0) node (1) [main node] {0};
  \draw (1,0) node (2) [main node] {1};
  
  \draw (.5,-1) node [label node] {$t = \pi/4$};
  \draw (.5,.75) node [label node] {$G_2$};
  
  \draw (1) -- (2);
\end{scope}

\begin{scope}[xshift=4.4cm]
  \draw (0,0) node (1) [main node] {0};
  \draw (1,0) node (2) [main node] {1};

  \draw (.5,-1) node [label node] {$t = 3\pi/2$};
  \draw (.5,.75) node [label node] {$G_3$};
  
  \path (2) edge [in=30, out=60, loop, above right, inner sep=0.02cm] (2);
\end{scope}

\begin{scope}[xshift=6.6cm]
  \draw (0,0) node (1) [main node] {0};
  \draw (1,0) node (2) [main node] {1};

  \draw (.5,-1) node [label node] {$t = \pi$};
  \draw (.5,.75) node [label node] {$G_4$};
  
    \path (2) edge [in=30, out=60, loop, above right, inner sep=0.02cm] (2);
\end{scope}

\begin{scope}[xshift=8.8cm]
  \draw (0,0) node (1) [main node] {0};
  \draw (1,0) node (2) [main node] {1};
  
  \draw (.5,-1) node [label node] {$t = 3\pi/2$};
  \draw (.5,.75) node [label node] {$G_5$};
  
 \path (2) edge [in=30, out=60, loop, above right, inner sep=0.02cm] (2);
\end{scope}

\begin{scope}[xshift=11.1cm]
  \draw (0,0) node (1) [main node] {0};
  \draw (1,0) node (2) [main node] {1};
  
  \draw (.5,-1) node [label node] {$t = \pi/4$};
  \draw (.5,.75) node [label node] {$G_6$};
  
  \draw (1) -- (2);
\end{scope}

\begin{scope}[xshift=13.3cm]
  \draw (0,0) node (1) [main node] {0};
  \draw (1,0) node (2) [main node] {1};
  
  \draw (.5,-1) node [label node] {$t = 3\pi/2$};
  \draw (.5,.75) node [label node] {$G_7$};
  
  \path (2) edge [in=30, out=60, loop, above right, inner sep=0.02cm] (2);
\end{scope}

\end{tikzpicture}
\caption{$HZH$}\label{fig:HZH}
\end{figure}

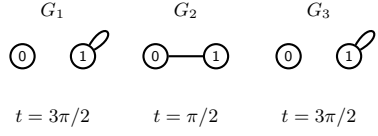
\begin{figure}
\centering
\begin{tikzpicture}[scale=0.8, every node/.style={transform shape}, thick, 
  main node/.style={draw, circle, inner sep=0.02cm, minimum size=0.4cm, font=\sffamily\scriptsize}, 
  label node/.style={inner sep=0cm}]

\begin{scope}[xshift=0cm]
  \draw (0,0) node (1) [main node] {0};
  \draw (1,0) node (2) [main node] {1};
  
  \draw (.5,-1) node [label node] {$t = 3\pi/2$};
  \draw (.5,.75) node [label node] {$G_1$};
  
  \path (2) edge [in=30, out=60, loop, above right, inner sep=0.02cm] (2);
\end{scope}

\begin{scope}[xshift=2.2cm]
  \draw (0,0) node (1) [main node] {0};
  \draw (1,0) node (2) [main node] {1};
  
  \draw (.5,-1) node [label node] {$t = \pi/2$};
  \draw (.5,.75) node [label node] {$G_2$};
  
  \draw (1) -- (2);
\end{scope}

\begin{scope}[xshift=4.4cm]
  \draw (0,0) node (1) [main node] {0};
  \draw (1,0) node (2) [main node] {1};
  
  \draw (.5,-1) node [label node] {$t = 3\pi/2$};
  \draw (.5,.75) node [label node] {$G_3$};
  
  \path (2) edge [in=30, out=60, loop, above right, inner sep=0.02cm] (2);
\end{scope}

\end{tikzpicture}
\caption{Simplified $HZH$, equivalent to $X$ gate.}\label{fig:newX}
\end{figure}

It is also worth noting that particular sequence of graphs can be combined into a much simpler single graph if they satisfy certain commutativity and edge properties.
\begin{obs}\label{obs:repeatedgecommute}
Let $G$ be a graph, $A_G$ its adjacency matrix, $\{H_i\}_{i=1}^n$ be subgraphs of $G$ that have adjacency matrices $\{A_i\}_{i=1}^n$ that a) all pairwise commute, b) $||A_i|| = ||A_j||$ for all $i \neq j$, and c) that each edge $ij \in E(G)$ shows up in exactly $\ell$ of the subgraphs $H_i$. Then 
\begin{equation*}
    (\prod_{j=1}^n e^{-i A_j t/ ||A_j||})^k = e^{-i A_G t'/ ||A_G||}
\end{equation*}
where $t' = kt \ell (||A_G||/ ||A_1||)$.

\end{obs}

\begin{proof}
    Since all $A_i$ pairwise commute and $||A_i|| = ||A_j||$ for all $i \neq j$,
    \begin{align*}
        (\prod_{j=1}^n e^{-i A_j t/||A_j||})^k &= e^{-i kt(A_1 + A_2 + \ldots + A_n)/ ||A_1||} \\
        & = e^{-i kt \ell A_G/ ||A_1||} = e^{-i kt \ell ||A_G|| A_G/ (||A_1|| ||A_G||)} = e^{-i kt \ell (||A_G||/ ||A_1||) A_G/ ||A_G||}.
    \end{align*}
    Setting $t' = kt \ell (||A_G||/ ||A_1||)$ gives the result.
\end{proof}

In particular, hypercubes and subcubes, as shown in Fig.~\ref{fig:subcubes}, satisfy this observation when $t' = 3t$. To see this, note that $||A_4|| = 3$ and $||A_1|| = ||A_2|| = ||A_3|| = 2$. It is easy to see from Fig.~\ref{fig:subcubes} that each edge in the right graph appears in exactly two of the three graphs in the left graphs. It is also easy to check that the adjacency matrices of the left three graphs all pairwise commute. Then,
\begin{align*}
         e^{-i A_1 t/2}e^{-i A_2 t/2} e^{-i A_3 t/2}&= e^{-i t(A_1 + A_2 + A_3)/2} \\
        & = e^{-i 2t A_G/ 2} = e^{-i 3t A_G/ 3} = e^{-i t' A_G/ ||A_G||}.
    \end{align*}

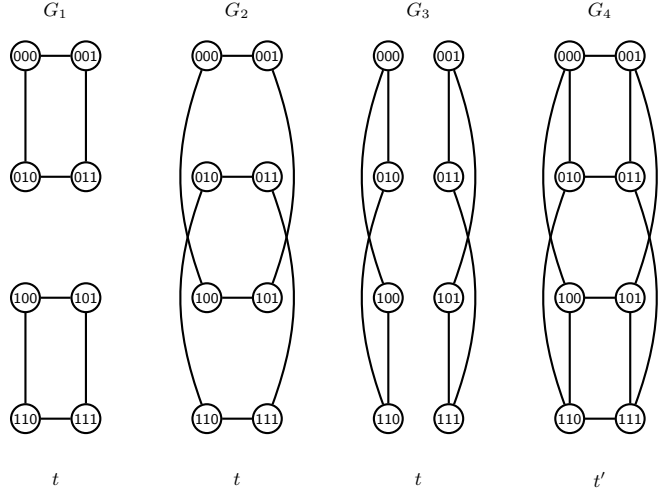
\begin{figure}
\centering
\begin{tikzpicture}[scale=0.8, every node/.style={transform shape}, thick, 
  main node/.style={draw, circle, inner sep=0.02cm, minimum size=0.4cm, font=\sffamily\scriptsize}, 
  label node/.style={inner sep=0cm}]

\begin{scope}[xshift=0cm]
  \draw (0,0) node (1) [main node] {000};
  \draw (1,0) node (2) [main node] {001};
  \draw (0,-2) node (3) [main node] {010};
  \draw (1,-2) node (4) [main node] {011};
   \draw (0,-4) node (5) [main node] {100};
  \draw (1,-4) node (6) [main node] {101};
  \draw (0,-6) node (7) [main node] {110};
  \draw (1,-6) node (8) [main node] {111};
  
  \draw (.5,-7) node [label node] {$t$};
  \draw (.5,.75) node [label node] {$G_1$};

    \draw (1) -- (2);
      \draw (3) -- (4);
        \draw (5) -- (6);
          \draw (7) -- (8);

    \draw (1) -- (3);
      \draw (2) -- (4);
        \draw (5) -- (7);
          \draw (6) -- (8);
\end{scope}

\begin{scope}
  \draw (3,0) node (1) [main node] {000};
  \draw (4,0) node (2) [main node] {001};
  \draw (3,-2) node (3) [main node] {010};
  \draw (4,-2) node (4) [main node] {011};
   \draw (3,-4) node (5) [main node] {100};
  \draw (4,-4) node (6) [main node] {101};
  \draw (3,-6) node (7) [main node] {110};
  \draw (4,-6) node (8) [main node] {111};
  
  \draw (3.5,-7) node [label node] {$t$};
  \draw (3.5,.75) node [label node] {$G_2$};

    \draw (1) -- (2);
      \draw (3) -- (4);
        \draw (5) -- (6);
          \draw (7) -- (8);

    \draw (1) to [in=110, out=250] (5);
      \draw (2) to [in=70, out=290] (6);
        \draw (3) to [in=110, out=250] (7);
          \draw (4) to [in=70, out=290] (8);
\end{scope}

\begin{scope}
  \draw (6,0) node (1) [main node] {000};
  \draw (7,0) node (2) [main node] {001};
  \draw (6,-2) node (3) [main node] {010};
  \draw (7,-2) node (4) [main node] {011};
   \draw (6,-4) node (5) [main node] {100};
  \draw (7,-4) node (6) [main node] {101};
  \draw (6,-6) node (7) [main node] {110};
  \draw (7,-6) node (8) [main node] {111};
  
  \draw (6.5,-7) node [label node] {$t$};
  \draw (6.5,.75) node [label node] {$G_3$};

    \draw (1) -- (3);
      \draw (2) -- (4);
        \draw (5) -- (7);
          \draw (6) -- (8);

    \draw (1) to [in=110, out=250] (5);
      \draw (2) to [in=70, out=290] (6);
        \draw (3) to [in=110, out=250] (7);
          \draw (4) to [in=70, out=290] (8);
\end{scope}

\begin{scope}
  \draw (9,0) node (1) [main node] {000};
  \draw (10,0) node (2) [main node] {001};
  \draw (9,-2) node (3) [main node] {010};
  \draw (10,-2) node (4) [main node] {011};
   \draw (9,-4) node (5) [main node] {100};
  \draw (10,-4) node (6) [main node] {101};
  \draw (9,-6) node (7) [main node] {110};
  \draw (10,-6) node (8) [main node] {111};
  
  \draw (9.5,-7) node [label node] {$t'$};
  \draw (9.5,.75) node [label node] {$G_4$};

    \draw (1) -- (3);
      \draw (2) -- (4);
        \draw (5) -- (7);
          \draw (6) -- (8);

    \draw (1) -- (2);
      \draw (3) -- (4);
        \draw (5) -- (6);
          \draw (7) -- (8);

    \draw (1) to [in=110, out=250] (5);
      \draw (2) to [in=70, out=290] (6);
        \draw (3) to [in=110, out=250] (7);
          \draw (4) to [in=70, out=290] (8);
\end{scope}

\end{tikzpicture}
\caption{Subcubes of $Q_3$ (left three graphs) and $Q_3$ (right graph).}\label{fig:subcubes}
\end{figure}

The next observation allows us to pull self-loops through edges in the graph representation of the dynamic CTQW.

\begin{obs}\label{obs:looppassing}

Let $A_1$ be a two-vertex graph with a self-loop on the vertex labeled by $\ket{1}$, $A_2$ be a two-vertex graph with an edge connecting $\ket{0}$ and $\ket{1}$, and $A_3$ be a two-vertex graph with a self-loop on the vertex labeled by $\ket{0}$. Then the CTQW on the dynamic graph sequence $\{(A_1, t), (A_2, \frac{(2m+1)\pi}{2})\}$ is equivalent to the CTQW on the dynamic graph sequence  $\{(A_2, \frac{(2m+1)\pi}{2}), (A_3, t)\} $. 
\end{obs}

\begin{proof}
First, let us describe the CTQW on the dynamic graph sequence $\{(A_1, t), (A_2, \frac{(2m+1)\pi}{2})\}$. Note that
\begin{align*}
    e^{-i A_1 t} &= \ketbra{0}{0} + e^{-it}\ketbra{1}{1} 
\end{align*}
and
\begin{align*}
     e^{-i A_2 \frac{(2m+1)\pi}{2}} &= \cos{\frac{(2m+1)\pi}{2}} (\ketbra{0}{0} + \ketbra{1}{1}) - i \sin{\frac{(2m+1)\pi}{2}} (\ketbra{0}{1} + \ketbra{1}{0})\\
     &= -i (-1)^m(\ketbra{0}{1} + \ketbra{1}{0}).
\end{align*}

Multiplying these together yields
\begin{equation*}
(-i (-1)^m(\ketbra{0}{1} + \ketbra{1}{0}))(\ketbra{0}{0} + e^{-it}\ketbra{1}{1}) =  -i(-1)^m (e^{-it} \ketbra{0}{1} + \ketbra{1}{0}).
\end{equation*}

Now, let us consider the second dynamic graph sequence. Note that
    \begin{equation*}
    e^{-i A_3 t} = e^{-it}\ketbra{0}{0} + \ketbra{1}{1}. 
    \end{equation*}

    Multiplying the two matrices together yields

    \begin{align*}
 (e^{-it}\ketbra{0}{0} + \ketbra{1}{1})(-i (-1)^m(\ketbra{0}{1} + \ketbra{1}{0})) =  -i(-1)^m (e^{-it} \ketbra{0}{1} + \ketbra{1}{0}).       
    \end{align*}
\end{proof}

The next observation allows one to absorb select phases, which are represented by self-loops in the dynamic CTQW model, into edges, which represent X rotations.

\begin{obs}\label{obs:edgeconjugation}
    The CTQW on the dynamic graph sequence $\{(A_1, t), (A_2, (2k+1)\pi), (A_1, t)\}$, where $A_1$ is the adjacency matrix of a two-vertex graph with a single edge and $A_2 $ is the adjacency matrix of the two-vertex graph with a single self-loop on either vertex $\ket{0}$ or vertex $\ket{1}$, is equivalent to the dynamic graph sequence $\{(A_2, (2k+1)\pi)\}$. Thus, we can conjugate self-loops by edges, on either vertex.
\end{obs}

\begin{proof}
Note that the dynamic graph sequence $\{(A_2, (2k+1)\pi)\}$ is equivalent to the Z gate, so we need only show a CTQW on the first dynamic graph sequence is also a Z gate. Consider $\{(A_1, t), (A_2, (2k+1)\pi), (A_1, t)\}$. Note that

\begin{align*}
     e^{-i A_1 t} &= \cos{t} (\ketbra{0}{0} + \ketbra{1}{1}) - i \sin{t} (\ketbra{0}{1} + \ketbra{1}{0})
\end{align*}
and
\begin{align*}
    e^{-i A_2 (2k+1)\pi} &= \pm \ketbra{0}{0} \mp \ketbra{1}{1} = \pm Z
\end{align*}

Then, 
{\begin{align}
    \big( \cos{t} \; \mathbb{I} - i \; \sin{t} \; X \big) \big( \pm Z \big) \big( \cos{t} \; \mathbb{I} - i \; \sin{t} \; X \big) = \big( \cos{t} \; \mathbb{I} - i \; \sin{t} \; X \big) \big( \cos{t} \; \mathbb{I} + i \; \sin{t} \; X \big) \big( \pm Z \big) = \pm Z.
\end{align}}
\end{proof}

\begin{obs}\label{obs:loopedgeshift}
Let $A_1$ be the adjacency matrix of a $2^n$-vertex graph with self-loops on vertices $j$ and $k$, and let and $A_2$ be the adjacency matrix of a $2^n$-vertex graph with the single edge $jk$. For any $t \in \mathbb{R}$, the CTQW on the dynamic graph sequence $\{(A_1, t), (A_2, \pi)\}$ is equivalent to the CTQW on the dynamic graph sequence $\{(A_1, t + \pi)\}$.
\end{obs}

\begin{proof}
The adjacency matrices are
\begin{align*}
A_1 \;&=\; \ketbra{j}{j} + \ketbra{k}{k}, \\
A_2 \;&=\; \ketbra{j}{k} + \ketbra{k}{j}.
\end{align*}
Let us consider $e^{-i A_2 \pi}e^{-i A_1 t}$,
\begin{align*}
    e^{-i A_2 \pi}e^{-i A_1 t} &= [-(\ketbra{j}{j} + \ketbra{k}{k}) + \sum_{\ell \notin \{j,k\}} \ketbra{\ell}{\ell}][e^{-it}(\ketbra{j}{k}+\ketbra{k}{j})+ \sum_{\ell \notin \{j,k\}} \ketbra{\ell}{\ell}] = \\
 &  \sum_{\ell \notin \{j,k\}} \ketbra{\ell}{\ell} -e^{-i t}[\ketbra{j}{k} + \ketbra{k}{j}]].
\end{align*}
Now, consider $e^{-i A_1 (t+\pi)}$
\begin{align*}
    e^{-i A_1 (t+\pi)} &= (\cos(t+ \pi)-i \sin(t + \pi)) \ketbra{j}{k}+\ketbra{k}{j}+ \sum_{\ell \notin \{j,k\}} \ketbra{\ell}{\ell}] = \\
 &   (-\cos(t)+i \sin(t)) \ketbra{j}{k}+\ketbra{k}{j}+ \sum_{\ell \notin \{j,k\}} \ketbra{\ell}{\ell}] = \\
 &-e^{-i t} [\ketbra{j}{k}+\ketbra{k}{j}]+ \sum_{\ell \notin \{j,k\}} \ketbra{\ell}{\ell}.
\end{align*}
\end{proof}

\section{Conclusion}\label{sec:discussion}
 In this work, we introduce five new dynamic graph sequence simplification techniques, as well as formalize the dynamic CTQW equivalent to the $U^{1/n}$ gate for a general unitary $U$. These simplifications include rules for passing graphs that only have edges through graphs that only have loops or vice versa, and for combining commuting graphs. We note that this list of graph simplification rules may not be exhaustive. Avenues for future research include finding new simplification rules for CTQWs on directed graph and developing simplification rules for dynamic CTQWs on directed graphs, which have been implemented and have applications including PageRank \cite{wang2020experimental}.

Continuous-time quantum walks on graphs are universal for computation \cite{childs2003universal, herrman2019continuous}, however little research has studied how one can convert back and forth between the CTQW model of computing and the gate model. This work, along with other dynamic CTQW papers \cite{herrman2019continuous, herrman2022relating, wong2019isolated, adisa2021implementing} give direct mappings between select sets of gates in the gate model of quantum computing and dynamic CTQWs. Furthermore, as evidenced in \cite{gonzales2025efficient} and \cite{atallah2026simulating}, CTQWs on dynamic graphs can be used to implement deterministic state preparation and Hamiltonian simulation algorithms that require fewer CX gates than other approaches. Thus, other avenues for future research include developing new transpiler routines or AI models that can decide which form of computation (gate model or dynamic CTQWs converted to the gate model) is best suited for different applications, where best suited may mean requiring fewer controlled gates or shorter circuits. 

\section*{Author contributions}
M.A. and J.M. developed graph simplification rules. 
D.D. helped derive the CTQW gates. 
Z.S. secured funding.
R.H. developed graph simplification rules and secured funding. All authors wrote, read, edited, and approved the final manuscript.

\section*{Acknowledgments}
M. Atallah, R. Herrman, and J. Mahmud acknowledge DE-SC0024290. D. Dilley and Z. Saleem acknowledge DOE-145-SE-14055-CTQW-FY23. The funder played no role in study design, data collection, analysis and interpretation of data, or the writing of this manuscript. 

\section*{Competing interests}
All authors declare no financial or non-financial competing interests. 

\section*{Code and Data Availability}
 \label{sec:codeAndDataAvail}
 Data sharing not applicable to this article as no datasets were generated
or analyzed during this study.

\vspace{20pt}
\noindent
\framebox{\parbox{\linewidth}{
The submitted manuscript has been created by UChicago Argonne, LLC, Operator of 
Argonne National Laboratory (``Argonne''). Argonne, a U.S.\ Department of 
Energy Office of Science laboratory, is operated under Contract No.\ 
DE-AC02-06CH11357. 
The U.S.\ Government retains for itself, and others acting on its behalf, a 
paid-up nonexclusive, irrevocable worldwide license in said article to 
reproduce, prepare derivative works, distribute copies to the public, and 
perform publicly and display publicly, by or on behalf of the Government.  The 
Department of Energy will provide public access to these results of federally 
sponsored research in accordance with the DOE Public Access Plan. 
http://energy.gov/downloads/doe-public-access-plan.}}

\bibliographystyle{unsrt}
\bibliography{refs}

\end{document}